\documentclass[12pt]{article}

\usepackage[margin=1.2in]{geometry}

\usepackage{amsmath}
\usepackage{amssymb}
\usepackage{amsthm}
\usepackage{xspace}
\usepackage{microtype}
\usepackage{mathtools}
\usepackage{esvect}
\usepackage{cases}
\usepackage[T1]{fontenc}
\usepackage{authblk}
\usepackage{url}
\usepackage{enumerate}

\DeclarePairedDelimiter\ceil{\Big\lceil}{\Big\rceil}
\DeclarePairedDelimiter\abs{\lvert}{\rvert}

\newlength{\arrow}
\newcommand*{\goesto}[1]{\xrightarrow{\mathmakebox[\arrow]{#1}}}
\newcommand*{\D}[1]{\frac{d#1}{dt}}
\newcommand*{\ie}{i.e.\@\xspace}
\newcommand*{\eg}{e.g.\@\xspace}
\newcommand*{\br}{\vec{r}}
\newcommand*{\bp}{\vec{p}}
\newcommand*{\bx}{\vec{x}}
\newcommand*{\bz}{\vec{z}}
\newcommand*{\by}{\vec{y}}
\newcommand*{\bxn}{\vec{x}_0}
\newcommand*{\N}{\mathbb{N}}
\newcommand*{\R}{\mathbb{R}}
\newcommand*{\Q}{\mathbb{Q}}
\newcommand*{\Z}{\mathbb{Z}}
\newcommand*{\C}{\mathbb{C}}
\newcommand*{\td}{[0,\infty)}
\renewcommand*{\vec}[1]{\mathbf{#1}}

\numberwithin{equation}{section}
\newtheorem{construction}{Construction}[section]
\newtheorem{theorem}{Theorem}[section]
\newtheorem{lemma}[theorem]{Lemma}
\newtheorem{corollary}[theorem]{Corollary}

\renewenvironment{proof}
{{\noindent \it Proof.}}
{\qed}

\begin{document}

    \title{
        Real-Time Computability of Real Numbers by Chemical Reaction Networks
        \thanks{
            This research was supported in part by National Science Foundation Grants 1247051 and 1545028.
            A preliminary version of a portion of this work was presented at the Sixteenth International Conference on Unconventional Computation and Natural Computation (UCNC 2017, Fayetteville, AR, June 5-9, 2017).
        }}

    \author[1]{Xiang Huang}
	\author[2]{Titus H. Klinge}
	\author[1]{James I. Lathrop}
	\author[1]{Xiaoyuan Li}
	\author[1]{Jack~H.~Lutz}

	\affil[1]{Department of Computer Science, Iowa State University, Ames, IA 50012, USA\\
	{\{huangx,jil,forbesii,lutz\}@iastate.edu}}

	\affil[2]{Department of Computer Science, Grinnell College, Grinnell, IA 50112, USA\\
	klingeti@grinnell.edu}

	\date{}

    \maketitle

    \begin{abstract}
        We explore the class of real numbers that are computed in real time by deterministic chemical reaction networks that are \emph{integral} in the sense that all their reaction rate constants are positive integers.
        We say that such a reaction network \emph{computes} a real number $\alpha$ in \emph{real time} if it has a designated species $X$ such that, when all species concentrations are set to zero at time $t = 0$, the concentration $x(t)$ of $X$ is within $2^{-t}$ of $\abs{\alpha}$ at all times $t \ge 1$, and the concentrations of all other species are bounded.
        We show that every algebraic number and some transcendental numbers are real time computable by chemical reaction networks in this sense.
        We discuss possible implications of this for the 1965 Hartmanis-Stearns conjecture, which says that no irrational algebraic number is real time computable by a Turing machine.
    \end{abstract}

	\section{Introduction}
Chemical reaction networks, originally conceived as descriptive mathematical models of molecular interactions in well-mixed solutions, are also widely used as prescriptive mathematical models for engineering molecular processes.
In the present century this prescriptive use of chemical reaction networks has been automated by software compilers that translate chemical reaction networks into complete specifications of DNA strand displacement systems that simulate them \cite{jSoSeWi10,jCDSPCS13}.
Chemical reaction networks have thus become the programming language of choice for many molecular programming applications.

There are several alternative semantics (operational meanings, also called kinetics) for chemical reaction networks.
The two oldest and most widely used of these are \emph{deterministic mass-action semantics} and \emph{stochastic mass-action semantics}.
This paper concerns the former of these, so for the rest of this paper, a \emph{chemical reaction network} (briefly, a \emph{CRN} or a \emph{deterministic CRN}) is a chemical reaction network with deterministic mass-action semantics.
This model is precisely specified in section~2 below.
For this introduction, it suffices to say that such a CRN is an ordered pair $N =(S, R)$, where $S$ is a finite set of \emph{species} (abstract molecule types), and $R$ is a finite set of \emph{reactions}, each of which has some form like: 
\[
X + Z \goesto{k} 2Y + Z,
\]
where $X, Y, Z \in S$ are species and $k \in (0, \infty)$ is a \emph{rate constant}.
A \emph{state} $\bx$ of $N$ specifies the real-valued \emph{concentration} $\bx(Y) \in \td$ of each species $Y \in S$.
Given an initial state $\bx(0)$ at time $t = 0$, deterministic mass action semantics specify the (continuous) evolution of the state $\bx(t)$ over time.

Even prior to the implementation of chemical reaction networks as a programming language it was clear that they constitute a model of computation.
In the case of deterministic CRNs, Stansifer had reportedly proven \cite{oCSWB09,oSolo16} that this model is Turing universal, \ie, that every algorithm can be simulated by a deterministic CRN, but no proof was published.
(Note: The title of \cite{jMagn97} seems to make this assertion, but the paper only exhibits a way to use deterministic CRNs to simulate finite Boolean circuits.) Fages, Le Guludec, Bournez, and Pouly \cite{cFLBP17} have now proven this universality theorem.

Deterministic chemical reaction networks are an analog model of computation, both in the intuitive sense that their states are vectors of real-valued concentrations that vary continuously over real-valued times and in the technical sense that they are a special case of Shannon's \emph{general purpose analog computer} (\emph{GPAC}) \cite{jShan41}, as explained in section 5 below.

This paper studies the ability of deterministic CRNs to rapidly compute real numbers in the following analog sense.
We say that a deterministic CRN \emph{computes} a real number $\alpha$ \emph{in real time} if it has a designated species $X$ such that the following three things hold.
(See section~3 for more details.)
First, the CRN's reaction rate constants are positive integers, and it is initialized with all concentrations set to zero at time $t=0$.
This implies that the CRN is, like any reasonable model of computation, finitely specifiable.
It also implies that only countably many real numbers are real time CRN-computable.
Second, there is some fixed bound on all the CRN's concentrations.
Under deterministic mass-action semantics, this implies that all the reaction rates of the CRN are bounded, whence time is a meaningful resource.
Third, the concentration $x(t)$ of the designated species $X(t)$ is within $2^{-t}$ of $\abs{\alpha}$ -- \ie, within $t$ bits of accuracy of $\abs{\alpha}$ -- at all times $t \ge 1$.
We say that the real number $\alpha$ is \emph{real time computable by chemical reaction networks} (briefly, \emph{real time CRN-computable}) if there is a CRN that computes $\alpha$ in this sense.
Elementary properties of real-time CRN computability are developed in section~3.

Our main theorem says that every algebraic number (\ie, every real solution of a polynomial with integer coefficients) is real time CRN-computable.
This result is proven in sections~4 and 5.
We prove in Section~6 that some transcendental (\ie, non-algebraic) real numbers are also real time CRN-computable.

Section~7 contains two discussions.
First, we compare real-time CRN computability with computability in the closely related \emph{large population protocol (LPP)} model of Bournez, Fraigniaud, and Koegler \cite{cBoFrKo12,oKoeg12}.
Second, our main theorem is a counterpoint -- but not a disproof -- of the 57-year-old, open Hartmanis-Stearns conjecture that no algebraic irrational is real time computable by a Turing machine \cite{jHarSte65}.
We discuss this contrast in some detail and pose two questions whose answers would shed further light on the computational complexities of algebraic irrationals.

    \section{Chemical Reaction Networks}
A \emph{species} is an abstract type of molecule.
Capital Roman characters such as $X$, $Y$, and $Z$ are commonly used to distinguish different species, but we also use decorations such as $X_0$, $\widehat{Y}$, and $\overline{Z}$ to distinguish them.

A \emph{reaction} over a finite set $S$ of species is a tuple $\rho=(\br,\bp,k)\in\N^S\times\N^S\times(0,\infty)$ and its components are called the \emph{reactant vector}, the \emph{product vector}, and the \emph{rate constant}, respectively.
(Here $\N^S$ denotes the set of all functions mapping $S$ into $\N$.)
To avoid excessive use of subscripts, for a reaction $\rho$ we use $\br(\rho)$, $\bp(\rho)$, and $k(\rho)$ to access the individual components of $\rho$.
A species $Y\in S$ is called a \emph{reactant} if $\br(Y)>0$, called a \emph{product} if $\bp(Y)>0$, and called a \emph{catalyst} if $\br(Y)=\bp(Y)>0$.
The \emph{net effect} of reaction $\rho=(\br,\bp,k)$ is the vector $\Delta\rho\in\N^S$ defined by
\[
	\Delta\rho(Y) = \bp(Y) - \br(Y)
\]
for each $Y\in S$.

A \emph{chemical reaction network} (\emph{CRN}) is an ordered pair $N=(S,R)$ where $S$ is a finite set of species and $R$ is a finite set of reactions over $S$.
Although this completes the definition of the \emph{syntax} of a CRN, we have yet to define the \emph{semantics} used in this paper.

Under \emph{deterministic mass action semantics}, the \emph{state} of a CRN $N=(S,R)$ at time $t$ is a real-valued vector $\bx(t)\in\td^S$, and for $Y\in S$, we call $\bx(t)(Y)$ the \emph{concentration} of $Y$ in $\bx(t)$.
We also write $y(t)=\bx(t)(Y)$ to denote the concentration of species $Y$ at time $t$.

The \emph{rate} of a reaction $\rho$ at time $t$ is defined as
\begin{equation}\label{eq:reaction_rate}
	\text{rate}_{\rho}(t) = k(\rho)\cdot\prod_{Y\in S}y(t)^{\br(\rho)(Y)}.
\end{equation}
This conforms to the so-called law of mass action which states that the rate of a reaction is proportional to the concentration of its reactants.

The \emph{total rate of change} of a species $Y\in S$ depends on the rates of all reactions in the CRN and the magnitude of their net effect on $Y$.
Therefore the concentration $y(t)$ conforms to the ordinary differential equation (ODE)
\begin{equation}\label{eq:total_rate_of_change}
	\D{y} = \sum_{\rho\in R}\Delta\rho(Y)\cdot\text{rate}_\rho(t).
\end{equation}

If we let $\mathcal{E}_Y$ be the ODE above for each $Y\in S$, then the \emph{mass action system} of the CRN is the coupled system
\begin{equation}\label{eq:mass-action-system}
	(\mathcal{E}_Y\mid Y\in S).
\end{equation}
Given an initial state $\bxn\in\td^S$, the behavior of the CRN is defined as the solution to the initial value problem (IVP) of the mass action system \eqref{eq:mass-action-system} along with the initial condition
\[
	y(0) = \bxn(Y)
\]
for each $Y\in S$.
    \section{Real-Time CRN Computability}
We say that a real number $\alpha$ is \emph{real time computable by chemical reaction networks} (briefly, \emph{real time CRN-computable}), and we write $\alpha \in \R_{RTCRN}$, if there exist a chemical reaction network $N = (S, R)$ and a species $X \in S$ with the following three properties:
\begin{enumerate}[1]
    \item (integrality).
    The CRN $N$ is \emph{integral} in the sense that:
    \begin{equation}\label{def:rt_integrality}
        k(\rho) \in \Z^+ 
    \end{equation}
    for all $\rho \in R$.

    \item (boundedness).
    There is a constant $\beta >0$ such that, if $N$ is initialized with $y(0) = 0$ for all $Y \in S$, then, for all $Y \in S$ and $t \in [0, \infty)$,
    \begin{equation}\label{def:rt_boundedness}
        y(t) \le \beta.
    \end{equation}

    \item (real-time convergence).
    If $N$ is initialized with $y(0)=0$ for all $Y \in S$, then for all $t \in [1, \infty)$,
    \begin{equation}\label{def:rt_convergence}
        \abs{x(t) - \abs{\alpha}} \le 2^{-t}.
    \end{equation}
\end{enumerate}

The integrality condition \eqref{def:rt_integrality} prevents the CRN $N$ from ``cheating'' by having information about $\alpha$ explicitly encoded into its rate constants.
To see that this is necessary to avoid nontriviality, note that, for any $\alpha\in(0, \infty)$, if the simple CRN:
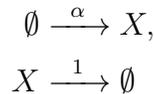
\begin{align*}
    \emptyset &\goesto{\alpha} X,\\
    X         &\goesto{1}      \emptyset
\end{align*}
is initialized with $x(0)=0$, then
\[
    x(t) = \alpha(1 - e^{-t})
\]
for all $t\in\td$.

The boundedness condition \eqref{def:rt_boundedness} imposes a ``speed limit'' on the CRN $N$.
This prevents $N$ from acting as a ``Zeno machine'' (machine that does infinite work in finite time) in the sense of Weyl \cite{oWeyl27}.
More precisely, condition \eqref{def:rt_boundedness} ensures that the reaction rates 
\eqref{eq:reaction_rate} of $N$ are all bounded.
This implies that the arc length of the curve traced by the state $\bx(s)$ of $N$ for $0 \le s \le t$ is $\theta(t)$, \ie, bounded above and below by positive constant multiples of $t$.
Pouly \cite{oPoul15,cBoGrPo16} has convincingly argued (in a more general setting) that this arc length, which we call the \emph{reaction clock time}, is the correct measure of the time that a CRN spends computing during the interval $[0,t]$.
Viewed in this light, condition \eqref{def:rt_boundedness} ensures that $t$ is, up to constant multiples, an accurate measure of the reaction clock time of $N$ during the interval $[0, t]$.

The real-time convergence condition \eqref{def:rt_convergence} requires the CRN $N$ to compute $\abs{\alpha}$ to within $t$ bits of accuracy by each time $t \ge 1$.
Note that this is an \emph{analog} approximation of $\abs{\alpha}$.
The CRN $N$ is not required to explicitly produce symbols in any sort of digital representation of $\abs{\alpha}$.

For the rest of this paper, unless otherwise noted, all CRNs $N=(S,R)$ are assumed to be initialized with $y(0)=0$ for all $Y\in S$.

To save space in our first lemma, we define the predicate
\begin{align*}
    \Phi_{\tau,\gamma}(\alpha) \equiv\;
        &\text{there exist a CRN }N=(S,R)\text{ and a species }X\in S\\
        &\text{satisfying \eqref{def:rt_integrality} and \eqref{def:rt_boundedness} such that, for all }t\in[\tau,\infty),\\
        &\abs{x(t)-\abs{\alpha}}\le e^{-\gamma t}
\end{align*}
for each $\tau,\gamma\in(0,\infty)$ and $\alpha\in\R$.
Note that $\Phi_{1,\ln 2}(\alpha)$ is the assertion that $\alpha\in\R_{RTCRN}$.
The following convenient lemma says that the definition of $\R_{RTCRN}$ is robust with respect to linear changes in condition \eqref{def:rt_boundedness}.

\begin{lemma}\label{lemma:real-time-dilation}
    For each $\alpha\in\R$ the following conditions are equivalent.
    \begin{enumerate}
        \item $\alpha\in\R_{RTCRN}$.
        \item There exists $\tau,\gamma\in(0,\infty)$ such that $\Phi_{\tau,\gamma}(\alpha)$ holds.
        \item For every $\tau,\gamma\in(0,\infty)$, $\Phi_{\tau,\gamma}(\alpha)$ holds.
    \end{enumerate}
\end{lemma}
\begin{proof}
    Let $\alpha\in\R$.
    It is clear that $(3) \Rightarrow (1) \Rightarrow (2)$, so it suffices to prove that $(2) \Rightarrow (3)$.
    For this, let $N, X, \tau, $ and $\gamma$ testify that $(2)$ holds, \ie, let $N$ and $X$ testify that $\Phi_{\tau, \gamma}(\alpha)$ holds.
    To prove $(3)$, let $\widehat{\tau}, \widehat{\gamma} \in (0, \infty)$.
    It suffices to show that  $\Phi_{\widehat{\tau}, \widehat{\gamma}}(\alpha)$ holds.
    Let
    \[
        a = \max\left\{\ceil{\frac{\tau}{\widehat{\tau}}}, \ceil{\frac{\widehat{\gamma}}{\gamma}}\right\},
    \]
    and let $\widehat{N} =(S, \widehat{R})$, where
    \[
        \widehat{R} = \{ (\br, \bp, ak) \mid (\br, \bp, k) \in R \}.
    \]
    That is, $\widehat{N}$ is exactly like $N$, except that each rate constant of $N$ has been multiplied by the positive integer $a$.
    Then $\widehat{N}$ is an integral CRN that is a ``sped up version'' of $N$ in the sense that, for all $y\in S$ and $t\in\td$,
    \begin{equation}\label{eq:time_dilation_equivalent}
        y_{\widehat{N}}(t) = y_{N}(at),
    \end{equation}
    where $y_{N}$ and $y_{\widehat{N}}$ are the values of $y$ in $N$ and $\widehat{N}$, respectively.
    This immediately implies that $\widehat{N}$ satisfies  \eqref{def:rt_boundedness}.
    Now let $t \in [\widehat{\tau}, \infty)$.
    Then $at \in [\tau, \infty)$, so our assumption $\Phi_{\tau, \gamma}(\alpha)$ tells us that 
    \begin{align*}
        \lvert x_{\widehat{N}}(t) - \abs{\alpha} \rvert
            &= \lvert x_{N}(at) - \abs{\alpha} \rvert \\
            &\le e^{-\gamma at}  \\
            &\le e^{-\widehat{\gamma} t},
    \end{align*}
    
    affirming $\Phi_{\widehat{\tau}, \widehat{\gamma}}(\alpha)$.
\end{proof}

The following lemma is a warm-up for our examination of $\R_{RTCRN}$

\begin{lemma}\label{lemma:rationals-rtcrn}
    $\Q\subsetneqq\R_{RTCRN}$
\end{lemma}
\begin{proof}
    If $\alpha=0$, then the CRN $N=(\{X\},\emptyset)$ testifies that $\alpha\in\R_{RTCRN}$.
    If $\alpha\in\Q\setminus\{0\}$, then we can write $\abs{\alpha}=\frac{a}{b}$, where $a, b \in \Z^+$.
    Then the integral CRN
    \begin{align*}
        \emptyset &\goesto{a} X \\
        X &\goesto{b} \emptyset
    \end{align*}
    satisfies 
    \[
        x(t) = \frac{a}{b}(1-e^{-bt}),
    \]
    so $\alpha \in \R_{RTCRN}$ by Lemma~\ref{lemma:real-time-dilation}.
    This shows that $\Q \subseteq  \R_{RTCRN}$.

    To see that $\Q \neq \R_{RTCRN}$, it suffices to show that $\frac{1}{\sqrt{2}} \in \R_{RTCRN}$.
    Since the integral CRN
    \begin{align*}
        \emptyset &\goesto{1} X\\
        2X &\goesto{2} X
    \end{align*}
    satisfies
    \[
        x(t) = \frac{1}{\sqrt{2}}\left(\frac{1-e^{-2\sqrt{2}t}}{1+e^{-2\sqrt{2}t}}\right),
    \]
    we have that
    \begin{align*}
        \abs*{x(t)-\frac{1}{\sqrt{2}}}
            &= \frac{1}{\sqrt{2}}\left(\frac{e^{-2\sqrt{2}t}}{1+e^{-2\sqrt{2}t}}\right)\\
            &\le \frac{1}{\sqrt{2}}e^{-2\sqrt{2}t}<e^{-2\sqrt{2}t},
    \end{align*}
    so $\frac{1}{\sqrt{2}}\in\R_{RTCRN}$ by Lemma~\ref{lemma:real-time-dilation}.
\end{proof}

Computable real numbers were introduced by Turing \cite{jTuri36,jTuri37} and have been extensively investigated \cite{oKo91,oWeih00}.

A real number $\alpha$ is \emph{computable}, and we write $\alpha \in \R_{comp}$, if there is a computable function $\widehat{\alpha} : \N \rightarrow \Q$ such that, for all $r \in \N$
\[
    \abs{\widehat{\alpha}(r) - \alpha} \leq 2^{-r}.
\]

\begin{lemma}
    $\R_{RTCRN} \subsetneqq \R_{comp}$
\end{lemma}
\begin{proof}
   Let $\alpha \in \R_{RTCRN}$, and let $N = (S, R)$ and $X \in S$ testify to this fact. Let $Y_1, \ldots, Y_n$ be the distinct species in $S$. Then the ODEs \eqref{eq:total_rate_of_change} can be written in the form 
   \begin{gather}
    y_1' = f_1(y_1,\ldots,y_n),\nonumber\\
    \vdots \label{eq:system-of-odes-1}\\
    y_n' = f_n(y_1,\ldots,y_n),\nonumber
    \end{gather}
    where $f_1,\ldots,f_n$ are polynomials with integer coefficients. By the boundedness condition \eqref{def:rt_boundedness} and Theorem 16 of \cite{jGraPou16}, the solution $\by: [0, \infty) \to [0, \infty)^n$ of \eqref{eq:system-of-odes-1} is polynomial time computable. It follows by the real-time convergence condition \eqref{def:rt_convergence} that $\alpha$ is computable in polynomial time in the sense of Ko \cite{oKo91}. Hence, $\alpha \in \R_{comp}$.
    
    It is well known \cite{oKo91} that not every computable real is computable in polynomial time, so the preceding paragraph proves the lemma.
\end{proof}

\section{Lyapunov CRN Computability}
This section defines a subclass of $\R_{RTCRN}$, namely, the class $\R_{LCRN}$ of all Lyapunov CRN-computable real numbers.
The main theorem of this section is the fact that $\R_{LCRN}$ is a field. 

Our definition of $\R_{LCRN}$ uses the stability theory of ordinary differential equations.
We review the elements of this theory that we need here, referring the reader to standard textbooks (\eg, \cite{oHirSma74,oTesc12}) for more thorough treatments.

We first note that the ordinary differential equations \eqref{eq:total_rate_of_change} of a CRN $N=(S,R)$ are \emph{autonomous}, meaning that they only depend on the time $t$ via the species concentrations $y(t)$.
Hence, if we let $Y_1,\ldots,Y_n$ be the distinct species in $S$, then the ODEs \eqref{eq:total_rate_of_change} can be written as
\begin{gather}
    y_1' = f_1(y_1,\ldots,y_n),\nonumber\\
    \vdots\label{eq:system-of-odes}\\
    y_n' = f_n(y_1,\ldots,y_n),\nonumber
\end{gather}
where $f_1,\ldots,f_n:\R^n\rightarrow\R$ are polynomials.
If we let $\vec{f}_N:\R^n\rightarrow\R^n$ be the function whose components are $f_1,\ldots,f_n$, then \eqref{eq:system-of-odes} can be written in the vector form
\begin{equation}
    \bx' = \vec{f}_N(\bx).
\end{equation}

A \emph{fixed point} of the CRN $N$ is a state $\bz\in\td^S$ such that $\mathbf{f}_N(\bz)=0$.
A state $\bz$ of $N$ is \emph{exponentially stable} if there exist $\alpha,\delta,C\in(0,\infty)$ such that, for all $\bxn\in\td^S$ with $\abs{\bxn-\bz}\le\delta$, if $N$ is initialized with $\bx(0)=\bxn$, then, for all $t\in\td$,
\begin{equation}
    \abs{\bx(t)-\bz}\le C e^{-\alpha t}\abs{\bx(0)-\bz}.
\end{equation}
It is easy to see that an exponentially stable state of $N$ must be a fixed point of $N$.

In this paper, we define a real number $\alpha$ to be \emph{Lyapunov computable by chemical reaction networks} (briefly, \emph{Lyapunov CRN-computable}), and we write $\alpha\in\R_{LCRN}$, if there exist a CRN $N=(S,R)$, a species $X\in S$, and a state $\bz\in\td^S$ with $\bz(X)=\abs{\alpha}$ such that the following conditions hold.
\begin{enumerate}[1]
    \item (integrality).
    The CRN $N$ is integral as in \eqref{def:rt_integrality}.

    \item (boundedness).
    Concentrations are bounded as in \eqref{def:rt_boundedness}.

    \item (exponential stability).
    $\bz$ is an exponentially stable state of $N$.

    \item (convergence).
    If $\bx(t) \in \td^S$ is the state of $N$ at time $t$, then
    \[
        \bz(0)=\mathbf{0} \Longrightarrow \lim_{t \rightarrow\infty} \bx(t) = \bz.
    \]
    (Here $\mathbf{0}$ is the state of $N$ defined by $\mathbf{0}(Y)=0$ for all $Y \in S$.)
   \end{enumerate}
   
A well known matrix characterization of exponential stability is useful for investigating the set $R_{LCRN}$.

The \emph{Jacobian matrix} of the CRN $N$ is the Jacobian matrix of $\mathbf{f}_N$, \ie, the $n\times n$ matrix
\[
    J_N = 
    \begin{pmatrix}
        \frac{\partial f_1}{\partial y_1} & \cdots & \frac{\partial f_1}{\partial y_n} \\
        \vdots  & \ddots & \vdots  \\
        \frac{\partial f_n}{\partial y_1} & \cdots & \frac{\partial f_n}{\partial y_n}
    \end{pmatrix}.
\]
More precisely, the \emph{Jacobian matrix} of $N$ \emph{in a state} $\bx\in\td^S$ is the matrix $J_N(\bx)$ in which each of the partial derivatives in $J_N$ is evaluated at the point $\bx$.
The \emph{eigenvalues} of the CRN $N$ \emph{in a state} $\bx\in\td^S$ are the eigenvalues of the matrix $J_N(\bx)$, \ie, the numbers $\lambda\in\C$ for which there exists $\by\in\R^n$ such that $J_N(\bx)(\by)=\lambda\by$.

Lyapunov's \emph{exponential stability theorem}, specialized to CRNs, says the following. 

\begin{theorem}\cite{oLyap92,oHirSma74}
    A fixed point $\bz$ of a CRN $N$ is exponentially stable if and only if all the eigenvalues of $N$ in state $\bz$ have negative real parts. 
\end{theorem}

We note the following easy inclusions.

\begin{lemma}\label{lemma:q_in_lcrn_in_rtcrn}
    $\Q\subseteq\R_{LCRN}\subseteq\R_{RTCRN}$.
\end{lemma}
\begin{proof}
    To see that $\Q\subseteq\R_{LCRN}$, it suffices to show that $\td\cap\Q\subseteq\R_{LCRN}$.
    For this, let $\alpha\in\td\cap\Q$.
    If $\alpha = 0$, then a one-species CRN with no reactions trivially affirms that $\alpha\in\R_{LCRN}$, so assume that $\alpha > 0$.
    Then there exist $p, q \in \Z^+$ such that $\alpha = \frac{p}{q}$.
    The CRN 
    \begin{align*}
        \emptyset &\goesto{p} X\\
        X         &\goesto{q} \emptyset
    \end{align*}
    computes $\alpha$ with species $X$.
    The ODE for this CRN is
    \[
        x'=p-qx
    \] 
    and the solution for the initial value $x(0)=0$ is $ x(t)=-\frac{p}{q}e^{-qt}+\frac{p}{q}$.
    We thus have $\lim_{t\to\infty}x(t)=\alpha$.
    Moreover, since we have an ODE system with only one variable, the eigenvalue at the fixed point is simply the derivative of $p-qx$, \ie, $-q$.
    Hence $\alpha\in \R_{LCRN}$.
    
    The inclusion $\R_{LCRN}\subseteq\R_{RTCRN}$ is immediate from the definitions of these classes and Lemma~\ref{lemma:real-time-dilation}.
\end{proof}

The rest of this section proves that $\R_{LCRN}$ has the closure properties of a field.

\begin{lemma}\label{lemma:lcrn_addition}
    For all $\alpha,\beta\in\td\cap\R_{LCRN}$, we have $\alpha+\beta\in\R_{LCRN}$.
\end{lemma}
\begin{proof}
    Let $\alpha,\beta\in\td\cap\R_{LCRN}$, and let $N_1= (S_1, R_1)$ and $N_2= (S_2, R_2)$ be CRNs that testify to this by computing $\alpha$ and $\beta$ with species $X$ and $Y$, respectively.
    We also let $\bz_1\in\td^{S_1}$ and $\bz_2\in\td^{S_2}$ be the equilibrium points that $N_1$ and $N_2$ use to compute $\alpha$ and $\beta$, \ie, $\bz_1(X)=\alpha$ and $\bz_2(Y)=\beta$.

    Let $N=(S,R)$ be the CRN defined by $S=S_1\cup S_2\cup S_3$ and $R=R_1\cup R_2\cup R_3$ where $S_3=\{U\}$ and $R_3$ consists of the three reactions 
    \begin{align*}
        X &\goesto{1} X + U\\
        Y &\goesto{1} Y + U\\
        U &\goesto{1} \emptyset.
    \end{align*}
    Note that the ODE for $U$ is
    \[
        \D{u}=x+y-u,
    \] 
    and the solution for $u(t)$ with all species initialized to zero is
    \[
        u(t)=ce^{-t}+e^{-t}\int_{1}^{t}e^s(x(s)+y(s))ds,
    \] 
    for some constant $c$.
    Hence,
    \[
        \lim_{t \to \infty}u(t)
            = 0+\lim_{t \to \infty}\frac{\int_{1}^{t}e^s(x(s)+y(s))ds}{e^t}.
    \]
    By L'H\^{o}pital's rule, we have
    \[
        \lim_{t \to \infty}u(t)
            = \lim_{t \to \infty}[x(t)+y(t)]=\alpha+\beta.
    \]

    It remains to be shown that the equilibrium point $(\bz_1,\bz_2,\alpha+\beta)$ is exponentially stable.
    First, we fix an order $S_1$, $S_2$, $S_3$ of the species in the Jacobian matrix $J_N$.
    We use $J_{i, j}$ to denote submatrix of $J_N$ that contains the partial derivatives of each species $A\in S_{i}$ with respect to each species $B\in S_{j}$ for $i,j\in\{1,2,3\}$.
    Then $J_N$ can be written as follows.
    \[
        J_{N} =
            \begin{pmatrix}
                J_{N_1} & \vec{0} & \vec{0} \\
                \vec{0} & J_{N_2} & \vec{0} \\
                J_{3, 1} & J_{3, 2} & J_{3, 3}
            \end{pmatrix}                  
    \]
    Since $N_1$ and $N_2$ have disjoint species, it is clear that $J_{1, 2}=\vec{0}$ and $J_{2, 1}=\vec{0}$.
    Furthermore, $N_1$ and $N_2$ are unaffected by the species $Y$, so $J_{1,3}=\vec{0}$ and $J_{2, 3}=\vec{0}$.
    We also note that $J_{3, 3}$ contains one element $\frac{\partial}{\partial u}(x+y-u) = -1$. 

    Since $J_N$ is a lower triangular block matrix,
    \begin{equation}\label{eq:det}
        \lvert J - \lambda I \rvert = 
            \lvert J_{N_{1}} - \lambda I_1 \rvert  \lvert J_{N_{2}} - \lambda I_2 \rvert  \lvert -1 - \lambda \rvert
    \end{equation}
    We can now conclude that the eigenvalues of $J_{N}$ are
    \begin{enumerate}[1)]
        \item the eigenvalues of $J_{N_{1}}$,
        \item the eigenvalues of $J_{N_{2}}$, and
        \item the eigenvalue of $J_{3, 3}$, which is $-1$.
    \end{enumerate}
    Since $\alpha, \beta$ are both in $\R_{LCRN}$, the real part of the eigenvalues of $J_{N_1}$ and  $J_{N_2}$ are all negative.
    Thus all the eigenvalues of $J_N$ have negative real parts.
    Therefore, $\alpha+\beta \in \R_{LCRN}$. 
\end{proof}

\begin{lemma}\label{lemma:lcrn_reciprocal}
    For each $0<\alpha\in\R_{LCRN}$, we have $\frac{1}{\alpha}\in\R_{LCRN}$.
\end{lemma}
\begin{proof}
    Let $\alpha\in(0,\infty)\cap\R_{LCRN}$, and let $N_1= (S_1, R_1)$ be a CRN that testifies to this by computing $\alpha$ with species $X$.
    We also let $\bz_1\in\td^{S_1}$be the equilibrium point that $N_1$ uses to compute $\alpha$, \ie, $\bz_1(X)=\alpha$.

    Let $N=(S,R)$ be the CRN with $S=S_1\cup S_2$ and $R=R_1\cup R_2$ where $S_1=\{Y\}$ and $R_2$ consists of the reactions
    \begin{align*}
        \emptyset &\goesto{1} Y\\
        X + Y     &\goesto{1} X
    \end{align*}
    so that
    \[
        \D{y} =  1 - xy.
    \]
    It is easy to see that the state $\bz=\left(\bz_1,\frac{1}{\alpha}\right)$ is the only reachable equilibrium point of $N$ from the initial state $\vec{0}$.
    We also note that the Jacobian matrix $J_N$ can be written as
    \[
        J_N =
            \begin{pmatrix}
                J_{N_1} & \vec{0} \\
                -y & -x
            \end{pmatrix}.
    \]
    We can see that $J_N(\bz)$ has eigenvalues with negative real parts, since the eigenvalues of $J_{N_1}$ have negative real parts and $-\alpha<0$.
    This implies that $\bz$ is exponentially stable, and therefore $\frac{1}{\alpha}\in\R_{LCRN}$.
\end{proof}

\begin{lemma}\label{lemma:lcrn_subtraction}
    For each $\alpha,\beta\in\R_{LCRN}$ with $\alpha\ge\beta\ge 0$, we have $\alpha-\beta\in\R_{LCRN}$.
\end{lemma}
\begin{proof}
    Let $\alpha$ and $\beta$ be as given.
    If $\beta=0$, then $\alpha-\beta=\alpha\in\R_{LCRN}$ by hypothesis.
    If $\beta=\alpha$, then $\alpha-\beta\in\R_{LCRN}$ by Lemma~\ref{lemma:q_in_lcrn_in_rtcrn}.
    Assume, then, that $\alpha>\beta>0$, and let $N_1=(S_1,R_1)$ and $N_2=(S_2,R_2)$ testify that $\alpha,\beta\in\R_{LCRN}$ using species $X_1\in S_1$ and $X_2\in S_2$, respectively.
    We also let $\bz_1$ and $\bz_2$ be the fixed points in $N_1$ and $N_2$ that compute $\alpha$ and $\beta$, respectively.

    Let $N=(S,R)$ be the CRN defined by $S=S_1\cup S_2\cup\{Y\}$ and $R=R_1\cup R_2\cup R_3$ where $R_3$ consists of the three reactions
    \begin{align*}
        \emptyset &\goesto{1} Y\\
        X_1 + Y   &\goesto{1} X_1\\
        X_2 + Y   &\goesto{1} X_2 + 2Y.
    \end{align*}
    Note that the additional reactions do not affect the species in $S_1$ and $S_2$, and yield the following ODE for $Y$.
    \[
        \D{y} = 1 - (x_1 - x_2)y.
    \]
    If $N$ is initialized to the $\vec{0}$ state, then by the construction of $N$, $\lim_{t\to\infty}(x_1(t)-x_2(t))=\alpha-\beta$.
    It is then easy to show that $\bz=(\bz_1,\bz_2,\frac{1}{\alpha-\beta})$ is a fixed point of $N$ and that $\lim_{t\to\infty}y(t)=\frac{1}{\alpha-\beta}$.
    The Jacobian matrix $J_N$ also has the form,
    \[
        J_{N} =
            \begin{pmatrix}
                J_{N_1} & \vec{0} & \vec{0} \\
                \vec{0} & J_{N_2} & \vec{0} \\
                J_{3, 1} & J_{3, 2} & -(x_1-x_2)
            \end{pmatrix}                  
    \]
    and therefore $J_N(\bz)$ has all negative eigenvalues since $J_{N_1}$ and $J_{N_2}$ have negative eigenvalues by the hypothesis and the eigenvalue $-(\alpha-\beta)$ is negative.
    Thus, $\frac{1}{\alpha-\beta}\in\R_{LCRN}$.

    Finally, by Lemma~\ref{lemma:lcrn_reciprocal} we can conclude that $\alpha-\beta\in\R_{LCRN}$.
\end{proof}

\begin{corollary}\label{lemma:lcrn_add_subgroup}
    $\R_{LCRN}$ is an additive subgroup of $\R$.
\end{corollary}
\begin{proof}
    Lemma~\ref{lemma:q_in_lcrn_in_rtcrn} tells us that $0\in\R_{LCRN}$, and the definition of $\R_{LCRN}$ implies that it is closed under additive inverses.
    To see that $\R_{LCRN}$ is closed under addition, let $\alpha,\beta\in\R_{LCRN}$.
    Then either
    \[
        \abs{\alpha + \beta} = \abs{\alpha} + \abs{\beta},
    \]
    in which case $\alpha + \beta \in \R_{LCRN}$ by Lemma~\ref{lemma:lcrn_addition}, or
    \[
        \abs{\alpha + \beta} = \max\{\abs{\alpha}, \abs{\beta}\} - \min\{\abs{\alpha}, \abs{\beta}\},
    \]
    in which case  $\alpha + \beta \in \R_{LCRN}$ by Lemma~\ref{lemma:lcrn_subtraction}.
\end{proof}

\begin{lemma}\label{lemma:lcrn_multiplication}
    For each $\alpha,\beta\in\td\cap\R_{LCRN}$, we have $\alpha\beta\in\R_{LCRN}$.
\end{lemma}
\begin{proof}
    Let $\alpha,\beta\in\td\cap\R_{LCRN}$, and let $N_1= (S_1, R_1)$ and $N_2= (S_2, R_2)$ be CRNs that testify to this by computing $\alpha$ and $\beta$ with species $X$ and $Y$, respectively.
    We also let $\bz_1\in\td^{S_1}$ and $\bz_2\in\td^{S_2}$ be the equilibrium points that $N_1$ and $N_2$ use to compute $\alpha$ and $\beta$, \ie, $\bz_1(X)=\alpha$ and $\bz_2(Y)=\beta$.

    Let $N=(S,R)$ be the CRN defined by $S=S_1\cup S_2\cup S_3$ and $R=R_1\cup R_2\cup R_3$ where $S_3=\{U\}$ and $R_3$ consists of the three reactions 
    \begin{align*}
        X + Y &\goesto{1} X + Y + U\\
        U &\goesto{1} \emptyset.
    \end{align*}
    Note that the ODE for $U$ is
    \[
        \D{u}=xy-u,
    \] 
    and the solution for $u(t)$ with all species initialized to zero is
    \[
        u(t)=ce^{-t}+e^{-t}\int_{1}^{t}e^sx(s)y(s)ds,
    \] 
    for some constant $c$.
    Hence,
    \[
        \lim_{t \to \infty}u(t)
            = 0+\lim_{t \to \infty}\frac{\int_{1}^{t}e^sx(s)y(s)ds}{e^t}.
    \]
    By L'H\^{o}pital's rule, we have
    \[
        \lim_{t \to \infty}u(t)
            = \lim_{t \to \infty}x(t)y(t)=\alpha\beta.
    \]

    It remains to be shown that the equilibrium point $(\bz_1,\bz_2,\alpha\beta)$ is exponentially stable.
    First, we fix an order $S_1$, $S_2$, $S_3$ of the species in the Jacobian matrix $J_N$.
    We use $J_{i, j}$ to denote submatrix of $J_N$ that contains the partial derivatives of each species $A\in S_{i}$ with respect to each species $B\in S_{j}$ for $i,j\in\{1,2,3\}$.
    Then $J_N$ can be written as follows.
    \[
        J_{N} =
            \begin{pmatrix}
                J_{N_1} & \vec{0} & \vec{0} \\
                \vec{0} & J_{N_2} & \vec{0} \\
                J_{3, 1} & J_{3, 2} & J_{3, 3}
            \end{pmatrix}                  
    \]
    Since $N_1$ and $N_2$ have disjoint species, it is clear that $J_{1, 2}=\vec{0}$ and $J_{2, 1}=\vec{0}$.
    Furthermore, $N_1$ and $N_2$ are unaffected by the species $Y$, so $J_{1,3}=\vec{0}$ and $J_{2, 3}=\vec{0}$.
    We also note that $J_{3, 3}$ contains one element $\frac{\partial}{\partial u}(xy-u) = -1$. 

    Since $J_N$ is a lower triangular block matrix,
    \begin{equation}\label{eq:det}
        \lvert J - \lambda I \rvert = 
            \lvert J_{N_{1}} - \lambda I_1 \rvert  \lvert J_{N_{2}} - \lambda I_2 \rvert  \lvert -1 - \lambda \rvert.
    \end{equation}
    We can now conclude the eigenvalues of $J_{N}$ are
    \begin{enumerate}[1)]
        \item the eigenvalues of $J_{N_{1}}$,
        \item the eigenvalues of $J_{N_{2}}$, and
        \item the eigenvalue of $J_{3, 3}$, which is $-1$.
    \end{enumerate}
    Since $\alpha, \beta$ are both in $\R_{LCRN}$, the real part of the eigenvalues of $J_{N_1}$ and  $J_{N_2}$ are all negative.
    Thus all the eigenvalues of $J_N$ have negative real parts.
    Therefore, $\alpha\beta \in \R_{LCRN}$. 
\end{proof}  

\begin{corollary}\label{lemma:lcrn_mult_subgroup}
    $\R_{LCRN}\setminus\{0\}$ is a multiplicative subgroup of $\R\setminus\{0\}$.
\end{corollary}
\begin{proof}
    Lemma~\ref{lemma:q_in_lcrn_in_rtcrn} tells us that $1 \in \R_{LCRN}$, so this follows immediately, from Lemmas~\ref{lemma:lcrn_reciprocal} and~\ref{lemma:lcrn_multiplication}.  
\end{proof}

We now have the main result of this section.
\begin{theorem}
    $\R_{LCRN}$ is a subfield of $\R$.
\end{theorem}
\begin{proof}
   This follows immediately from Corollaries~\ref{lemma:lcrn_add_subgroup} and~\ref{lemma:lcrn_mult_subgroup}.
\end{proof}
    
\section{Algebraic Numbers Are Lyapunov CRN-Computable}
In this section, we prove that every algebraic number is Lyapunov CRN-computable.
We begin by proving that algebraic numbers that are the smallest positive root of a polynomial with distinct roots are Lyapunov CRN-computable.
In this case, we construct a CRN with one species that when initialized to zero asymptotically approaches the smallest positive root of the polynomial.
We also ensure that the root is an exponentially stable equilibrium point of the CRN.

For an arbitrary algebraic number, we reduce the problem to the special case by shifting all the roots of its minimal polynomial by a rational number.
By doing so, we make the relevant root become the smallest positive root, and use the special case to complete the proof of the theorem.

\begin{lemma}\label{lemma:first_root_algebraic}
    Every algebraic number that is the smallest positive root of some integral polynomial with roots only of multiplicity one is in $\R_{LCRN}$.
\end{lemma}
\begin{proof}
    Let $\alpha>0$ be an algebraic number, and let $P(x)=c_n x^n+c_{n-1}x^{n-1}+\cdots + c_0$ be the polynomial with integral coefficients that testifies to this.
    Furthermore, we assume that $\alpha$ is the first positive root of $P$ and that the roots of $P$ only have multiplicity one. 
    Without loss of generality, we also assume that $c_0 \ge 0$, and use $-P(x)$ otherwise.
    We now build a CRN $N=(S,R)$ with species $S=\{X\}$ such that
    \[
        \D{x} = P(x).
    \]
    For each term $c_k x^k$ in $P(x)$, we include a reaction in $R$ to add this term to $\D{x}$.
    If $c_k>0$, we add the reaction
    \[
        k X \goesto{c_k} (k+1) X,
    \]
    and if $c_k<0$ we add the reaction
    \[
        k X \goesto{-c_k} (k-1) X.
    \]
    Note that it is possible for the number of products or reactants to be empty.
    For example, the reaction for $c_0$ is $\emptyset\goesto{c_0} X$ because $c_0>0$.
    Similarly, if $c_1<0$, the reaction would be $X\goesto{-c_1}\emptyset$.
    Also, if $c_k=0$, we do not add any reaction.

    It now suffices to show that $\alpha$ is an exponentially stable equilibrium point of $N$ and that if $X$ is initialized to zero, then $x(t)$ converges to $\alpha$.
    Since $c_0>0$, it is clear that $P(0)>0$, whence $\D{x}>0$ at time $0$.
    Finally, since $\alpha$ is the first positive root of $P(x)$ and $x(0)=0$, it is clear that $\lim_{t\rightarrow\infty} x(t)=\alpha$.

    To show that $\alpha$ is exponentially stable, it suffices to show that all eigenvalues of the Jacobian matrix at $\alpha$ have negative real parts, \ie, that $P'(\alpha)<0$.
    We show this using the following two facts.
    \begin{enumerate}
        \item
        $P(x)>0$ for all $x\in[0,\alpha)$.
        Otherwise $\alpha$ would not be the smallest positive root.

        \item
        $P'(\alpha) \neq 0$.
        Otherwise $\alpha$ would have multiplicity of at least two, but we assumed the roots of $P$ have multiplicity one.
    \end{enumerate}
    Finally, it is clear that
    \[
        P'(\alpha)=\lim_{x\to \alpha^-}\frac{P(x)-P(\alpha)}{x-\alpha} < 0,
    \]
    since $P(x)>0$ for all $x\in[0,\alpha)$ and $P'(x)\ne 0$.
\end{proof}    

\begin{theorem}\label{thm:lcrn_algebraic}
    Every algebraic number is an element of $\R_{LCRN}$.
\end{theorem}
\begin{proof}
    Let $\alpha>0$ be an algebraic number, and let $P(x)$ be the minimal polynomial with integral coefficients that testifies to this.
    Since $P$ is minimal, its roots have multiplicity one \cite{oGowe08}.
    Therefore, if $\alpha$ is the smallest positive root of $P$, then $\alpha\in\R_{LCRN}$ by Lemma~\ref{lemma:first_root_algebraic}.

    If $\alpha$ is not the smallest positive root of $P$, let $\beta$ be the largest positive root less than $\alpha$.
    Now let $\frac{p}{q}$ be a rational satisfying $\beta<\frac{p}{q}<\alpha$, and let $\gamma = \alpha-\frac{p}{q}$.
    If $n$ is the degree of $P$, then $P\left(x+\frac{p}{q}\right)\cdot q^n$ is an integral polynomial with distinct roots and $\gamma$ is its smallest positive root.
    By Lemma~\ref{lemma:first_root_algebraic}, $\gamma\in\R_{LCRN}$, and since $\Q\subseteq\R_{LCRN}$ and $\R_{LCRN}$ is closed under addition, $\frac{p}{q}+\gamma=\alpha\in\R_{LCRN}$.
\end{proof}

    \section{Real Time CRN-Computable Transcendental Numbers}
\newcommand{\troot}{\frac{e - 1 + \sqrt{(e-1)^2 +4}}{2}}

We now show that we can compute more than the algebraic numbers in real time by proving that a transcendental number can be computed in our model.
Our proof utilizes the following construction where the species $U$ is the one that computes the transcendental.

\begin{construction}
	Let $N = (S, R)$ where $S = \{U, V, X\}$ and $R$ consists of the following reactions.
	\begin{align*}
			\emptyset &\goesto{1} X \hspace*{0.5in}	&X &\goesto{1} \emptyset\\
			U         &\goesto{1} 2U &\emptyset &\goesto{1} U \\
			X + U     &\goesto{1} X &V &\goesto{1} 2V \\
			X         &\goesto{1} X + V 	&X + V &\goesto{1} X \\
			U + V     &\goesto{1} \emptyset  \\
	\end{align*}
\end{construction}

For the rest of this section, we assume that $N$ is initialized to the $\vec{0}$ state.
Since the ODE for the species $X$ is $\D{x}=1-x$, it is clear that for all time $t\in\td$, $x(t)=1-e^{-t}$.
The remainder of our proof depends on the following lemma which reduces the dynamical system into a single variable.

\begin{lemma}
	The ODE of the species $U$ can be written
	\begin{equation}
		\frac{du}{dt} = -u^2 + u f(t) + 1
	\end{equation}
	where $f(t) = e^{-t}+e^{1-e^{-t}} - 1$.
\end{lemma}
\begin{proof}
	Recall that $N$ is initialized to $\vec{0}$ and $x(t)=1-e^{-t}$.
	We also introduce the variable $y(t)=u(t)-v(t)$ for all $t\in\td$.
	Then by the construction of $N$,
	\[
		\D{y} = \D{u}-\D{v} = (u+1-xu-uv)-(v+x-xv-uv).
	\]
	Since $y = u-v$, the above ODE simplifies to
	\begin{equation}
		\D{y} = (y+1)(1-x).
	\end{equation}
	By separation of variables and integration, we conclude that $y(t)=e^{1-e^{-t}}-1$.
	Then $v=u-e^{1-e^{-t}}+1$, so it follows that
	\begin{align*}
		\frac{du}{dt}
			&= u + 1 - xu - uv \\
			&= u + 1 - (1-e^{-t})u - u(u-e^{1-e^{-t}}+1) \\
			&= -u^2 + u f(t) + 1.
	\end{align*}
\end{proof}

\begin{corollary}
	The ODE of the species $U$ can be written
	\[
		\frac{du}{dt} = -(u-r_1)(u-r_2)
	\]
	where $r_1(t) = \frac{f(t) + \sqrt{f(t)^2+4}}{2}$ and $r_2(t) = \frac{f(t) - \sqrt{f(t)^2+4}}{2}$.
\end{corollary}

To simplify our analysis, we begin with the following lemma which upper bounds the functions $r_1(t)$ and $u(t)-r_2(t)$.

\begin{lemma}
	For all $t\in\td$,
	\begin{align}
	 	u(t) - r_2(t)&\ge\sqrt{2} - 1\label{eq:u_less_than_constant}\\
	 	&\text{\hspace*{-2.0in}and} \notag\\
	 	r_1(t)&\ge\widehat{r}(t)\label{eq:r1_ge_rhat},
	 \end{align} 
	 where $\widehat{r}(t) = \frac{e-1+\sqrt{(e-1)^2+4}}{2}\left(1-e^{-t}\right)$
\end{lemma}

\begin{proof}
    Recall that $f(t)=e^{-t}+e^{1-e^{-t}} - 1$ is monotonically increasing wth upper bound $e-1$, and therefore $f(t)<2$ for all $t\in\td$.
    It follows that $r_2(t) > 1-\sqrt{2}$ for all $t\in\td$.
    Since $u(t)\ge 0$ for all $t\in\td$, we obtain inequality \eqref{eq:u_less_than_constant} from $u(t)-r_2(t)\ge 0-r_2(t) > \sqrt{2}-1$ for all $t\in\td$.

	It remains to be shown that \eqref{eq:r1_ge_rhat} holds, \ie, that
	\[
		\frac{f(t) + \sqrt{f(t)^2+4}}{2}
			\ge\frac{e-1+\sqrt{(e-1)^2+4}}{2}\left(1-e^{-t}\right),
	\]
	holds for all $t\in\td$.
	Since for all $t\in\td$
	\[
		f(t) = e^{-t} + e^{1-e^{-t}}-1 \ge (e-1)(1-e^{-t}),
	\]
	it suffices to show that $\sqrt{f(t)^2+4}\ge (e-1)^2+4\left(1-e^{-t}\right)$.
	This can be seen by
	\begin{align*}
		\sqrt{f(t)^2+4}
			&\geq \sqrt{(e-1)^2(1-e^{-t})^2+4}\\
			&\geq \sqrt{(e-1)^2(1-e^{-t})^2+4(1-e^{-t})^2} \\
			&= (1-e^{-t})\sqrt{(e-1)^2 + 4}.
	\end{align*}
\end{proof}

Noting that $u - r_2$ is always positive yields an immediate corollary of the above lemma that gives tight bounds on the behavior of $u(t)$.

\begin{corollary}\label{corrroot}
	For all $t\in\td$,
	\begin{equation}
		\widehat{u}(t) \le u(t) \le r_1(t)
	\end{equation}
	where $\widehat{u}(t)$ is the solution to the initial value problem of $\widehat{u}(0)=0$ and
	\[
		\D{\widehat{u}} = -(\widehat{u} - \widehat{r})(\sqrt{2}-1).
	\]
\end{corollary}

It is routine to check that the solution to $\widehat{u}(t)$ with $\widehat{u}(0)=0$ is 
\begin{equation}
	\widehat{u}(t) =
		\troot \left(1 - \frac{e^{-(\sqrt{2}-1)t}-(\sqrt{2}-1)e^{-t}}{1-(\sqrt{2}-1)}\right) \label{eq:root}
\end{equation}

\begin{lemma}\label{lemma_b}  $u(t)$ converges exponentially to
	\[
		\troot
	\]
	as $t\to\infty$.
\end{lemma}

\begin{proof}
	It is clear from \eqref{eq:root} and the definition of $r_1$ that $\widehat{u}(t)$ and $r_1(t)$ converge exponentially to $\troot$, so this follows from Corollary~\ref{corrroot}.
\end{proof}

\begin{theorem}\label{thm:rtcrn_has_trans}
	$\R_{RTCRN}$ contains a transcendental number.
\end{theorem}

\begin{proof}
	By Lemma~\ref{lemma_b}, it suffices to show that $\troot$ is transcendental.
	For this, let $x,y\in\R$ and assume that $y = x - 1+\sqrt{5-2x+x^2}$.
	It suffices to show that if $y$ is algebraic, then $x$ is algebraic.	
	Assume the hypothesis.
	Then $5-2x+x^2 = (y-x+1)^2$, so $x$ is the root of a polynomial, all of whose coefficients are in $\Z\cup\{y\}$, so $x$ is algebraic.
\end{proof}

    \section{Discussion}
We have shown that
\begin{align}
\textsc{Alg} \subseteq \R_{LCRN} \subseteq \R_{RTCRN}, \label{eq:inclusions}
\end{align}
where \textsc{Alg} is the field of algebraic numbers.  We have also shown that
\[
\textsc{Alg} \neq \R_{RTCRN},
\]
 so at least one of the two inclusions in \eqref{eq:inclusions} is proper.  At the time of this writing we do not know whether the left-hand inclusion is proper, and we do not know whether the right-hand inclusion is proper. 

Both notions of real time CRN computability discussed here, $\R_{RTCRN}$ and $\R_{LCRN}$, are closely related to the investigations by Bournez, Fraigniaud, and Koegler \cite{cBoFrKo12,oKoeg12} of computability by large population protocols.  Roughly speaking, a {\it large population protocol} ({\it LPP}) is a deterministic chemical reaction network in which every reaction has exactly two reactants and exactly two products.  Among other things, this implies that the sum of concentrations of all species is constant over time.  A real number $\alpha$ is defined to be {\it computable by an LPP} if there exist an LPP $N$, a state ${\bf z}$ of $N$, and a designated subset $D$ of the species of N with the following three properties.
\begin{enumerate}
\item[(1)] $N$ has only finitely many fixed points.
\item[(2)] $\bz$ is an exponentially stable state of $N$.  
\item[(3)] $\alpha$ is the sum of the concentrations of the species in $D$ in the state $\bz$.
\end{enumerate}

Bournez, Fraigniaud, and Koegler prove that a real number $\alpha$ is computable by an LPP if and only if $\alpha$ is an algebraic number.  The ``only if'' direction of their proof is an elimination of quantifiers argument \cite{bMark02} that depends crucially on (1) above.  It is to be hoped that further research will clarify the relationship between LPP computability and real time CRN computability.

What does \eqref{eq:inclusions} say about the complexity of algebraic irrationals on other models of computation?

The first thing to understand here is that deterministic chemical reaction networks are, in a very precise sense, a model of analog computation. In 1941, Shannon \cite{jShan41} introduced the \emph{general-purpose analog computer} (\emph{GPAC}). A GPAC is a mathematical abstraction of the \emph{differential analyzer}, an early analog computer that Bush \cite{jBush31} had constructed at MIT, and which Shannon had operated as a graduate research assistant. The GPAC model has been corrected and otherwise modified a number of times over the years \cite{jPouRic74,jLipRub87,jGraCos03,jGrac04}.
Its present form can be characterized in terms of circuits, but it is more simply characterized as a system 
\begin{equation}
    \by'(t)=\bp(t,y),
\end{equation}
of ordinary differential equations, where $\bp$ is a vector of polynomials.
A deterministic CRN is thus a special type of GPAC of the form
\begin{equation}
    \by'(t)=\bp(y),
\end{equation}
where each component $p_i$ of $\bp$ has the ``kinetic'' form $p_i(\by) = q_i(\by) - y_ir_i(\by)$, with $q_i$ and $r_i$ having nonnegative coefficients \cite{jHarTot81}. Our CRNs in this paper have the added constraints that all the coefficients in these polynomials are integers, and all concentrations are initialized to zero.  Our main theorem thus implies that all algebraic numbers are real time computable by GPACs that have only finite information coded into their parameters and initializations.

We now turn from analog computation to discrete computation. A famous conjecture of Hartmanis and Stearns \cite{jHarSte65} says that no irrational algebraic number is real time computable by a Turing machine.
This conjecture has been open for over 50 years.  Fischer, Meyer, and Rosenberg \cite{jFiMeRo70} proved that real-time computability on a Turing machine is equivalent to linear-time computability on a Turing machine.
Hence the Hartmanis-Stearns conjecture is equivalent to the statement that no irrational algebraic number is linear-time computable by a Turing machine. As observed by Gurevich and Shelah \cite{cGurShe89}, linear time is a very model-dependent notion. Hence, as stated, the Hartmanis-Stearns conjecture is a very specific conjecture about linear-time computation on Turing machines. 

Our main theorem does not disprove the Hartmanis-Stearns conjecture (nor was it intended to), but conceptually locating the gap between our main theorem and a disproof of the Hartmanis-Stearns conjecture would shed light on the computational complexities of algebraic irrationals.
This raises the following questions.

Question 1. Can CRNs in our model (or GPACs with only finite information encoded into their parameters and initializations) produce in linear time the individual digits of each real number that is real time CRN-computable? If so, our main theorem implies that the Hartmanis-Stearns conjecture fails for analog computation. If not, the Hartmanis-Stearns conjecture holds for analog computation and is essentially about producing the individual digits as opposed to the analog convergence that we have used here.

Question 2. Is there a reasonable discrete model of computation on which some algebraic irrational can be computed in linear time? If so, then the Hartmanis-Stearns conjecture is either false or model-dependent.
If not, then the Hartmanis-Stearns conjecture is true in a strong, model-independent way, at least for discrete computation. (Note that ``reasonable'' here excludes models that perform numerical operations faster than we know how to do them, because Brent \cite{jBren76} has shown how to compute $\sqrt{2}$ in linear time if integer multiplication can be done in linear time. See also \cite{oLip12}.)

    \section*{Acknowledgments}
We thank Olivier Bournez for telling us about the work [1,18] on large population protocols and for a discussion concerning the breakthrough [7] on Turing universality of deterministic CRNs.  We thank Wolfgang Kliemann, James Hanson, Jonathan Smith, and Matthew Riley for useful discussions, we thank Kendal Gast for useful suggestions on the exposition, and we thank Sergey Verlan for additional pointers to the literature.

    \bibliographystyle{plain}

\end{document}